\documentclass[a4paper,USenglish,cleveref,modulo,autoref,thm-restate]{lipics-v2021}

\graphicspath{{./figures/}}

\usepackage{wrapfig}
\usepackage{amssymb}
\usepackage{amsmath}
\usepackage{xspace}
\usepackage{thm-restate}
\usepackage[bottom]{footmisc}

\newcommand{\Reals}{\mathbb{R}}
\newcommand{\CH}{C\!H}

\newcommand{\etal}{\textit{et al.}\xspace}

\newcommand{\myparNS}[1]{\noindent{\bfseries #1.}}
\newcommand{\mypar}[1]{\medskip\myparNS{#1}}

\DeclareMathOperator{\sgn}{sgn}

\newtheorem{fact}[theorem]{Fact}

\title{Obstructing Classification via Projection}

\author{Pantea Haghighatkhah}{TU Eindhoven, The Netherlands}{p.haghighatkhah@tue.nl}{}{}
\author{Wouter Meulemans}
{TU Eindhoven, the Netherlands }{w.meulemans@tue.nl}{}{}
\author{Bettina Speckmann}{TU Eindhoven, the Netherlands }{b.speckmann@tue.nl}{https://orcid.org/0000-0002-8514-7858}{}
\author{J\'{e}r\^{o}me Urhausen}{Utrecht University, the Netherlands }{j.e.urhausen@uu.nl}{}{Supported by the Dutch Research Council (NWO); 612.001.651.}
\author{Kevin Verbeek}{TU Eindhoven, the Netherlands }{k.a.b.verbeek@tue.nl}{}{}

\authorrunning{P. Haghighatkhah, W. Meulemans, B. Speckmann, J. Urhausen, and K. Verbeek} 

\Copyright{Pantea Haghighatkhah, Wouter Meulemans, Bettina Speckmann, J\'{e}r\^{o}me Urhausen and Kevin Verbeek} 

\ccsdesc[500]{Theory of computation~Computational geometry}
\ccsdesc[300]{Theory of computation~Models of learning}

\keywords{Projection, classification, models of learning}

\acknowledgements{Research on the topic of this paper was initiated at the 5th Workshop on Applied Geometric Algorithms (AGA 2020) in Langbroek, NL. The authors thank Jordi Vermeulen for initial discussions on the topic of this paper.}

\EventEditors{John Q. Open and Joan R. Access}
\EventNoEds{2}
\EventLongTitle{42nd Conference on Very Important Topics (CVIT 2016)}
\EventShortTitle{CVIT 2016}
\EventAcronym{CVIT}
\EventYear{2016}
\EventDate{December 24--27, 2016}
\EventLocation{Little Whinging, United Kingdom}
\EventLogo{}
\SeriesVolume{42}
\ArticleNo{23}

\begin{document}

\hideLIPIcs
\nolinenumbers 
\maketitle              
\begin{abstract}
Machine learning and data mining techniques are effective tools to classify large amounts of data.
But they tend to preserve any inherent bias in the data, for example, with regards to gender or race. Removing such bias from data or the learned representations is quite challenging. In this paper we study a geometric problem which models a possible approach for bias removal. Our input is a set of points $P$ in Euclidean space $\Reals^d$ and each point is labeled with $k$ binary-valued properties. A priori we assume that it is ``easy'' to classify the data according to each property. Our goal is to obstruct the classification according to one property by a suitable projection to a lower-dimensional Euclidean space $\Reals^m$ ($m < d$), while classification according to all other properties remains easy. 

What it means for classification to be easy depends on the classification model used. 
We first consider classification by linear separability as employed by support vector machines. We use Kirchberger's Theorem to show that, under certain conditions, a simple projection to $\Reals^{d-1}$ suffices to eliminate the linear separability of one of the properties whilst maintaining the linear separability of the other properties. We also study the problem of maximizing the linear ``inseparability'' of the chosen property.
Second, we consider more complex forms of separability and prove a connection between the number of projections required to obstruct classification and the Helly-type properties of such separabilities.
\keywords{Projection \and classification \and models of learning.}
\end{abstract}

\section{Introduction}
\label{sec:introduction}

Classification is one of the most basic data analysis operators: given a (very) large set of high-dimensional input data with a possibly large set of heterogeneous properties, we would like to classify the data according to one or more of these properties to facilitate further analysis and decision making. Machine learning and data mining techniques are frequently employed in this setting, since they are effective tools to classify large datasets. However, such as any data-driven techniques, they tend to preserve any bias inherent in the data, for example, with regards to gender or race. Such bias arises from under-representation of minority groups in the data or is caused by historical data, which reflect outdated societal norms. Bias in the data might be inconsequential, for example in music recommendations, but it can be harmful when classification algorithms are used to make life-changing decisions on, for example, loans, recruitment, or parole~\cite{skeem2016risk}. 

Naturally, the identification and removal of bias receives a significant amount of attention, although the problem is still far from solved. For example, Mehrabi~\etal~\cite{DBLP:journals/corr/abs-1908-09635} provide a taxonomy of fairness definitions and bias types. They list the biases caused by data and the types of discrimination caused by machine learning techniques. Many approaches have been considered to eliminate or reduce bias in machine learning models. Some researchers have used a statistical approach to address this problem (e.g.,~\cite{DBLP:conf/nips/HardtPNS16}), while others focus on data preprocessing or controlling the sampling to compensate for bias or under-representation in the data (e.g.,~\cite{DBLP:conf/aies/AminiSSBR19,DBLP:journals/kais/KamiranC11}). Another approach is to use an additional (adversarial) machine learning model to eliminate bias in the first model (e.g.,~\cite{edwards2016censoring,pmlr-v80-madras18a,DBLP:conf/aies/ZhangLM18}). One major problem of attempting to eliminate bias (or increasing fairness) in machine learning is that it may negatively affect the accuracy of the learned model. This trade-off has also been studied extensively (e.g.,~\cite{DBLP:journals/corr/BerkHJJKMNR17,DBLP:conf/aistats/ZafarVGG17}).

We are particularly interested in data that is represented by vectors in high-dimensional Euclidean space. Such data arises, for example, from word embeddings for textual data. Several studies show that the bias present in the training corpora is also present in the learned representation (e.g.~\cite{DBLP:conf/icml/BrunetAAZ19,Caliskan183}). 
Abbasi~\etal~\cite{DBLP:conf/sdm/AbbasiFSV19} recently introduced a geometric notion of stereotyping. In this paper we follow the same premise that bias is in some form encoded in the geometric or topological features of the high-dimensional vector representation and that manipulating this geometry can remove the bias. This premise has been the basis for many papers on algorithmic fairness (e.g.,~\cite{edwards2016censoring,disparate,pmlr-v28-zemel13}).

Several papers investigate the theory that gender is captured in certain dimensions of the data. 
Bolukbasi~\etal~\cite{DBLP:conf/nips/BolukbasiCZSK16} postulate that the bias manifests itself in specific ``particularly gendered'' words and that equalizing distances to these special words removes bias.
Zhao~\etal~\cite{DBLP:conf/emnlp/ZhaoZLWC18} devise a model which attempts to represent gender in one dimension which can be removed after training to arrive at a (more) gender-neutral word representation.
Bordia and Bowman~\cite{DBLP:conf/naacl/BordiaB19} remove bias by minimizing the projection of the embeddings on the gender subspace (using a regularization term in the training process).
Very recently, various papers~\cite{Dev_Li_Phillips_Srikumar_2020,pmlr-v89-dev19a,Yuzi-AMC2020,DBLP:conf/acl/RavfogelEGTG20} explored the direct use of projection to remove sensitive properties of the data. In some cases the data is not projected completely, as removing sensitive properties completely may negatively affect the quality of the model. 

In this paper we take a slightly more general point of view. We say that a property is present in the data representation if it is ``easy'' to classify the data according to that property. That is, a property (such as gender) can be described by more complicated geometric relations than a subspace. Given the premise that the geometry of word embeddings encodes important relations between the data, then any bias removal technique needs to preserve as much as possible of these relations. Hence we investigate the use of projection to eliminate bias while maintaining as many other relations as possible. We say that the relation of data points with respect to specific properties is maintained by a projection, if it is still easy to classify according to these properties after projection. Our paper explores how well projection can obstruct classification according to a specific property (such as gender) for certain classification models.

\mypar{Problem statement}
Our input is a set of $n$ points $P = \{p_1, \ldots, p_n\}$ in general position in~$\Reals^d$. For convenience we identify the points with their corresponding vector. We model the various properties of the data (such as gender) as binary labels.\footnote{Neither gender nor many other societally relevant properties are binary, however, we restrict ourselves to binary properties to simplify our mathematical model.} Hence, for all points in~$P$ we are also given $k$ binary-valued \emph{properties}, represented as functions $a_i\colon P \rightarrow \{-1, 1\}$ for $1 \leq i \leq k$. 
We denote the subset of points $p \in P$ with $a_i(p) = 1$ as $P^i_{+}$, and the subset of points $p \in P$ with $a_i(p) = -1$ as $P^i_{-}$ for $1 \leq i \leq k$. For a point $p \in P$, we refer to the tuple $(a_1(p), \ldots, a_k(p))$ as the \emph{label} of $p$. Note that there are $2^k$ different possible labels. Generally speaking, we do not know which specific properties a dataset has. However, to study the influence of projection on all relevant properties of a dataset, we assume that these properties are given.

\begin{figure}[b]
    \centering
    \includegraphics{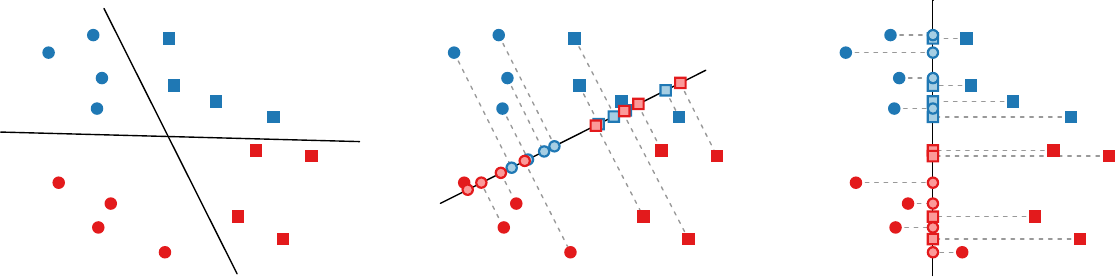}
    \caption{Left: data points with two linearly-separable properties: shape and color. Middle: a projection which keeps shape separated, but not color. Right: a projection with the opposite effect.}
    \label{fig:points_and_separators}
\end{figure}

We assume that it is ``easy'' to classify the points in $P$ according to the properties by using the point coordinates. Throughout the paper, we consider different definitions for what is considered easy or difficult to classify. Our goal is to compute a projection $P'$ of $P$ to lower dimensions such that the first property $a_1$ becomes difficult to classify in $P'$, and the other properties $a_2, \ldots, a_k$ remain easy to classify in $P'$. As a shorthand we use the notation $P_{-} = P^1_{-}$ and $P_{+} = P^1_{+}$ for the point sets in which the special property $a_1$ is set to $-1$ and $+1$, respectively. Similarly, we use the notation $P'_{-}$ and $P'_{+}$ for the point sets $P_{-}$ and $P_{+}$ after projection. In most cases we will consider a projection along a single unit vector $w$ ($\|w\| = 1$), mapping points in $\Reals^d$ to points in $\Reals^{d-1}$. For a point $p_i \in P$, we denote its projection as $p'_i = p_i - (p_i \cdot w) w$, where $(p_i \cdot w)$ denotes the dot product between the vectors $p_i, w \in \Reals^d$. 
To assign coordinates to $p'_i$ in $\Reals^{d-1}$, we need to establish a basis for the projected space. We therefore often consider $p'_i$ to lie in the original space $\Reals^d$, where the coordinates of $p'_i$ are restricted to the hyperplane that is orthogonal to $w$ and passes through the origin. Sometimes we will consider projections along multiple vectors $w_1, \ldots, w_r$. In that case we assume that $\{w_j\}_{j=1}^r$ form an orthonormal system, such that we can write the projection as $p'_i = p_i - \sum_{j=1}^r (p_i \cdot w_j) w_j$. Again, we assume that $p'_i$ still lies in $\Reals^d$, but is restricted to the $(d-r)$-dimensional flat that is orthogonal to $w_1, \ldots, w_r$ and passes through the origin. 

We consider different models for defining what is easy or difficult to classify, resulting in different computational problems. These models typically rely on a form of ``separability'' between two point sets. For a specific definition of separability, using a slight abuse of notation, we will often state that a property $a_i$ is separated in a point set $P$ when we actually mean that $P^i_{-}$ and $P_{+}^i$ are separated (see Figure~\ref{fig:points_and_separators} for a simple example in $\Reals^2$). The specific models, along with the relevant definitions, are described in detail in the respective sections.  

\mypar{Contributions and organization}
In Section~\ref{sec:LinearSep} we consider linear separability as the classification model. We first show that, if even one possible label is missing from $P$, then there may be no projection that eliminates the linear separability of $a_1$ whilst keeping the linear separability of the other properties. On the other hand, if all possible labels are present in the point set, then we show that it is always possible to achieve this goal. In Appendix~\ref{app:inseparable} we discuss a related question: given a measure to quantify how far removed a labeled point set is from linear separability, how can we optimize this measure for $a_1$ after projection? We show that the optimal projection can be computed efficiently under certain specific conditions, but may be hard to compute efficiently in general.
In Section~\ref{sec:GenSep} we introduce $(b, c)$-separability, which is a generalization of linear separability. Although a single projection is no longer sufficient to avoid $(b, c)$-separability of $a_1$ after projection, we show that, in general, the number of projections needed to achieve this is linked to the Helly number of the respective separability predicate. We then establish bounds on the Helly numbers of $(b,c)$-separability for specific values of $b$ and $c$. Omitted proofs can be found in Appendix~\ref{app:omitted}.

\section{Linear separability}
\label{sec:LinearSep}

In this section we consider linear separability for classification. For a point set $P$ and property $a_i\colon P \rightarrow \{-1, 1\}$, we say that $a_i$ is easy to classify on $P$ if $P^i_{-}$ and $P^i_{+}$ are (strictly) linearly separable; we say that $a_i$ is difficult to classify otherwise. Two point sets $P$ and $Q$ ($P, Q \subset \Reals^d$) are \emph{linearly separable} if there exists a hyperplane $H$ separating $P$ from $Q$. The point sets are \emph{strictly linearly separable} if we can additionally require that none of the points lie on $H$. Equivalently, the point sets $P$ and $Q$ are linearly separable if there exists a unit vector $v \in \Reals^d$ and constant $c \in \Reals$ such that $(v \cdot p) \leq c$ for all $p \in P$ and $(v \cdot q) \geq c$ for all $q \in Q$ ($v$ is the normal vector of the hyperplane $H$). We say that $P$ and $Q$ are linearly separable \emph{along} $v$. If the inequalities can be strict, then the point sets are strictly linearly separable.

One of the machine learning techniques that use linear separability for classification are \emph{support vector machines} (SVMs). SVMs compute the (optimal) hyperplane that separates two classes in the training data (if linearly separable), and use that hyperplane for further classifications. Linear separability is therefore a good first model to consider for classification.



Let $\CH(P)$ denote the convex hull of a point set $P$. By definition, we have that $x \in \CH(P)$ if and only if there exist coefficients $\lambda_i \geq 0$ such that $x = \sum_{i=1}^n \lambda_i p_i$ and $\sum_{i=1}^n \lambda_i = 1$. We use the following basic results on convex geometry and linear algebra.
\begin{fact}\label{fac:convexlinsep}
Two point sets $P$ and $Q$ are linearly separable iff $\CH(P)$ and $\CH(Q)$ are interior disjoint. $P$ and $Q$ are strictly linearly separable iff $\CH(P) \cap \CH(Q) = \emptyset$.
\end{fact}
\begin{observation}\label{obs:projectcoeff}
Let $P' = \{p'_1, \ldots, p'_n\}$ be the point set obtained from $P = \{p_1, \ldots, p_n\}$ by projecting along a unit vector $w$. If $x = \sum_{i = 1}^n \lambda_i p_i$ (for $\lambda_i \in \Reals$), then $x' = x - (w \cdot x) w = \sum_{i = 1}^n \lambda_i p'_i$. Specifically, if $x \in \CH(P)$, then $x' \in \CH(P')$.
\end{observation}
\begin{restatable}{lemma}{sepafterproj}
\label{lem:sepafterproj}
Let $P$ and $Q$ be two point sets. If we project both $P$ and $Q$ along a unit vector $w$ to obtain $P'$ and $Q'$, then $P'$ and $Q'$ are not strictly linearly separable iff there exists a line $\ell$ parallel to $w$ that intersects both $\CH(P)$ and $\CH(Q)$. If $\ell$ intersects the interior of $\CH(P)$ or $\CH(Q)$, then $P'$ and $Q'$ are not linearly separable.
\end{restatable}
\begin{proof}
Assume that the line $\ell$ exists, and it contains $x_P \in \CH(P)$ and $x_Q \in \CH(Q)$ (see Figure~\ref{fig:sepafterproj}). By construction, $x' = x_P - (w \cdot x_P) w = x_Q - (w \cdot x_Q) w$. Hence, by Observation~\ref{obs:projectcoeff}, $x' \in \CH(P') \cap \CH(Q')$. Thus, by  Fact~\ref{fac:convexlinsep}, $P'$ and $Q'$ are not strictly linearly separable. For the other direction, choose $x' \in \CH(P') \cap \CH(Q')$. The line parallel to $w$ and passing through $x'$ must clearly intersect both $\CH(P)$ and $\CH(Q)$. The extension to (non-strict) linear separability is straightforward.
\end{proof}

\begin{figure}[t]
    \centering
    \includegraphics{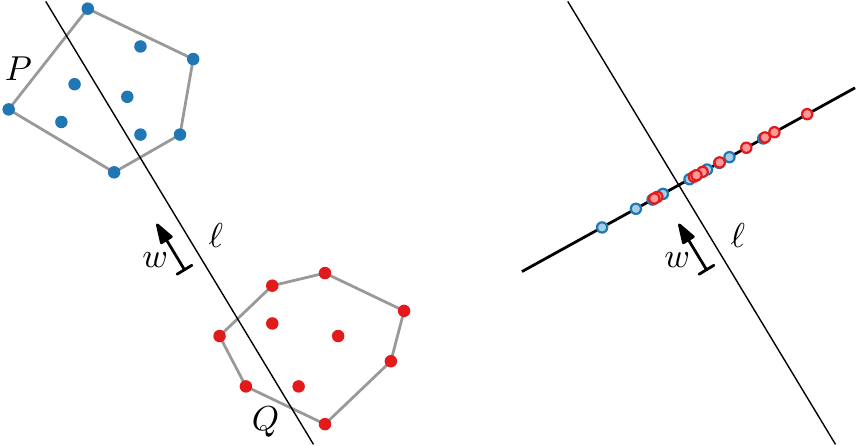}
    \caption{Line $\ell$ intersects $\CH(P)$ and $\CH(Q)$; after projection the convex hulls intersect.}
    \label{fig:sepafterproj}
\end{figure}
Assume now that the properties $a_1, \ldots, a_k$ are strictly linearly separable in $P$. Can we project $P$ along a unit vector $w$ so that $a_2, \ldots, a_k$ are still strictly linearly separable in $P'$, but $a_1$ is not? We consider two variants: (1) \emph{separation preserving} and (2) \emph{separability preserving} projections. The former preserves a fixed set of separating hyperplanes $H_2, \ldots, H_k$ for properties $a_2, \ldots, a_k$, the latter preserves only linear separability of $a_2, \ldots, a_k$.

Lemma~\ref{lem:not_all_labels} proves there exist point sets using only $2^k - 1$ possible labels for which every separability preserving projection also keeps $a_1$ strictly linearly separable after projection. The idea is to use the properties $a_2, \ldots, a_k$ to sufficiently restrict the direction of a separability preserving projection to make it impossible for this projection to eliminate the linear separability of $a_1$. A simple example for $d = k = 2$ is shown in Figure~\ref{fig:few_labels}.

\begin{figure}[h]
    \centering
    \includegraphics{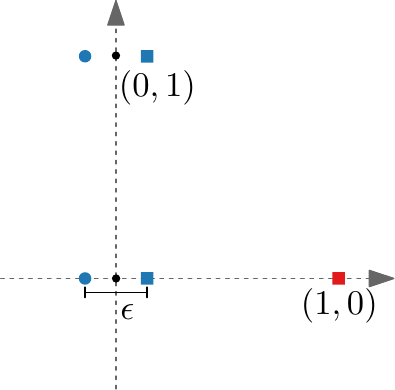}
    \caption{A point set with $5$ points and $2$ properties: $a_1$ (color) and $a_2$ (shape). To keep $a_2$ linearly separable after projection, the projection vector $w$ should be nearly vertical, but then $a_1$ will also remain linearly separable.}
    \label{fig:few_labels}
\end{figure}

\begin{restatable}{lemma}{notalllabels}
\label{lem:not_all_labels}
For all $k > 1$ and $d \geq k$, there exist point sets $P$ in $\Reals^d$ with properties $a_1, \ldots, a_k$ using $2^{k}-1$ labels such that any separability preserving projection along a unit vector $w$ also keeps $a_1$ strictly linearly separable after projection.
\end{restatable}

We now assume that all $2^k$ labels are used in $P$. Note that this assumption directly implies that $d \geq k$: take any set of $k$ separating hyperplanes $H_1, \ldots, H_k$ for the $k$ properties and consider the arrangement formed by the hyperplanes in $\Reals^d$. Clearly, all points in the same cell of the arrangement must have the same label. However, it is well-known that it is not possible to create $2^k$ cells in $\Reals^d$ with only $k$ hyperplanes if $d < k$. This has also interesting implications for the case when $d = k$: if we apply a separation preserving projection to $P$, then $a_1$ cannot be linearly separable in $P'$, since $P'$ is embedded in $\Reals^{k-1}$. 

We now show that, if $d \geq k$, then there always exists a separation preserving projection that eliminates the strict linear separability of $a_1$ (see  Figure~\ref{fig:lemma3}). Our proof uses Kirchberger's theorem~\cite{kirchberger1903}. Below we restate this theorem in our own notation. We also include our own proof, since the construction in the proof is necessary for efficient computation of our result.

\begin{restatable}[\cite{kirchberger1903}]{theorem}{kirchberger}
\label{thm:kirchberger}
Let $P$ and $Q$ be two points sets in $\Reals^d$ such that $\CH(P) \cap \CH(Q) \neq \emptyset$. Then there exist subsets $P^{*} \subseteq P$ and $Q^{*} \subseteq Q$ such that $\CH(P^{*}) \cap \CH(Q^{*}) \neq \emptyset$ and $|P^{*}| + |Q^{*}| = d+2$.
\end{restatable}
\begin{proof}
Let $|P| = n$ and $|Q| = m$. We show that, if $n + m \geq d+3$, then we can remove one of the points from either $P$ or $Q$. Pick a point $x \in \CH(P) \cap \CH(Q)$. By definition, we can find coefficients $\lambda_1, \ldots, \lambda_n \geq 0$ and $\mu_1, \ldots, \mu_m \geq 0$ such that $\sum_{i=1}^n \lambda_i p_i = x = \sum_{j=1}^m \mu_j q_j$, $\sum_{i=1}^n \lambda_i = 1$, and $\sum_{j=1}^m \mu_j = 1$. If any of these coefficients is zero, then we can remove the corresponding point whilst keeping $x$ in the intersection of the two convex hulls. Otherwise, we find nonzero coefficients $a_1, \ldots, a_n$ and $b_1, \ldots, b_m$ such that $\sum_{i=1}^n a_i p_i = \sum_{j=1}^m b_j q_j$, $\sum_{i=1}^n a_i = 0$, and $\sum_{j=1}^m b_j = 0$. As this is a linear system with $d+2$ constraints and $n+m \geq d+3$ variables, there must exist a set of nonzero coefficients that satisfy these constraints. Let $\rho_\lambda = \min\{\lambda_i/a_i\mid a_i > 0\}$, $\rho_\mu = \min \{\mu_j/b_j \mid b_j > 0\}$, and $\rho = \min(\rho_\lambda, \rho_\mu)$. Now consider the new coefficients $\lambda'_i = \lambda_i - \rho a_i$ and $\mu'_j = \mu_j - \rho b_j$. By construction we have that $\lambda'_i \geq 0$ for $1 \leq i \leq n$, $\mu'_j \geq 0$ for $1 \leq j \leq m$, $\sum_i \lambda'_i = \sum_j \mu'_j = 1$, and $\sum_i \lambda'_i p_i = \sum_j \mu'_j q_j = x'$. Additionally, one of the new coefficients is zero, and we can remove the corresponding point. We can repeat this process until $n+m = d+2$.
\end{proof}

\begin{figure}[t]
    \centering
    \includegraphics{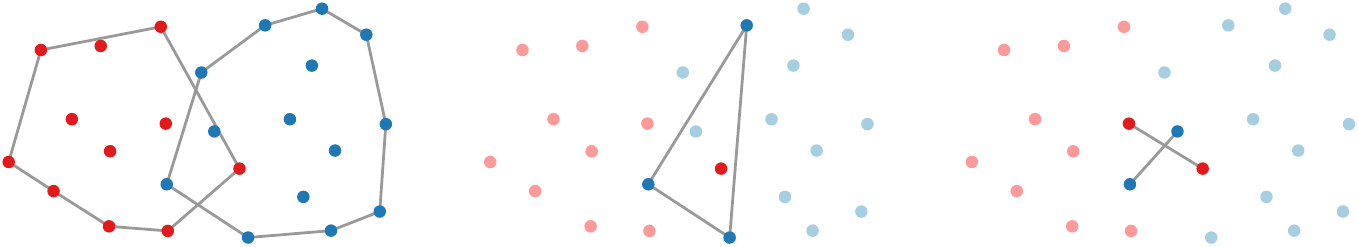}
    \caption{Theorem \ref{thm:kirchberger} in 2D: 4 points are needed to construct two intersecting convex hulls.}
    \label{fig:lemma3}
\end{figure}

\noindent The following proof constructs a suitable projection vector using four main steps:
\begin{enumerate}
    \item We project the points orthogonally onto the linear subspace $A$ spanned by the normals of the separating hyperplanes $H_2, \ldots, H_k$.
    \item We argue that, since $P$ uses all $2^k$ labels, $a_1$ is not linearly separable in $A$.
    \item We find a small subset of points $P^{*}$ for which $a_1$ is not linearly separable in $A$.
    \item We construct a separation preserving projection that maps all points in $P^{*}$ to an affine transformation of $A$. As a result, $a_1$ is not strictly linearly separable after projection.
\end{enumerate}
We assume that the points in $P$, along with the chosen separating hyperplanes, are in \emph{general position}. Specifically, we assume that any set of $d$ vectors, where each vector is either a distinct difference vector of two points in $P$ or the normal vector of one of the separating hyperplanes, is linearly independent. Note that, since all properties are initially strictly linearly separable, it is always possible to perturb the separating hyperplanes to ensure general position, assuming that $P$ is also in general position.  

\begin{theorem}\label{thm:linsep}
If $P$ is a point set in $\Reals^d$ in general position with $k \leq d$ strictly linearly separable properties $a_1, \ldots, a_k$ using all~$2^k$ labels, then there exists a separation preserving projection along a unit vector $w$ that eliminates the strict linear separability of $a_1$.
\end{theorem}
\begin{proof}
We provide an explicit construction of the vector $w$. Let $H_2, \ldots, H_k$ be any separating hyperplanes (in general position with $P$) for each of the properties $a_2, \ldots, a_k$ in $P$, respectively. Let $v_i$ be the normal of hyperplane $H_i$ for $2 \leq i \leq k$, and let $A \subset \Reals^d$ be the $(k-1)$-dimensional linear subspace spanned by $v_2, \ldots, v_k$. Furthermore, let $H^{*} = \bigcap_{i=2}^k H_i$ be the $(d-k+1)$-dimensional flat that is the intersection of the separating hyperplanes. Note that a projection along a vector $w$ is separation preserving if and only if $w$ is parallel to $H^{*}$. Let $T(p)$ be the result of an orthogonal projection of a point $p \in P$ onto $A$. For ease of argument, we also directly apply an affine transformation that maps $H^{*}$ (which intersects $A$ in one point by construction) to the origin, and maps $v_2, \ldots, v_k$ to the standard basis vectors of $\Reals^{k-1}$. 

Now define $Q_{-} = \{T(p)\mid p \in P_{-}\}$ and $Q_{+} = \{T(p)\mid p \in P_{+}\}$. By construction, since all labels are used by $P$, both $Q_{-}$ and $Q_{+}$ must have a point in each orthant of $\Reals^{k-1}$. If a point set $Q$ has a point in each orthant, then $\CH(Q)$ must contain the origin;
because if it does not, then there exists a vector $v$ such that $(v \cdot q) > 0$ for all $q \in Q$. But there must exist a point $q^{*} \in Q$ whose sign for each coordinate is opposite from that of $v$ (or zero), which means that $(v \cdot q^{*}) \leq 0$, a contradiction. Thus, both $\CH(Q_{-})$ and $\CH(Q_{+})$ contain the origin, and $\CH(Q_{-}) \cap \CH(Q_{+}) \neq \emptyset$. We now apply Theorem~\ref{thm:kirchberger} to $Q_{-}$ and $Q_{+}$ to obtain $Q_{-}^{*}$ and $Q_{+}^{*}$ consisting of $k+1$ points in total. Let $P^{*} \subseteq P$ be the corresponding set of original points that map to $Q_{-}^{*} \cup Q_{+}^{*}$. We can now construct $w$ as follows. Pick a point $p^{*} \in P^{*}$, and let $F_1$ be the unique $(k-1)$-dimensional flat that contains the remaining points in $P^{*}$. Let $F_2$ be the flat obtained by translating $H^{*}$ to contain $p^{*}$. Since $F_1$ is $(k-1)$-dimensional and $F_2$ is $(d - k + 1)$-dimensional, $F_1 \cap F_2$ consists of a single point $r \in \Reals^d$ (assuming general position). The desired projection vector is now simply $w = r - p^{*}$ (normalized if necessary).

We finally show that the constructed vector $w$ has the correct properties. First of all, $w$ is parallel to $H^{*}$ by construction, and hence the projection along $w$ is separation preserving. Second, since $r \in F_1$ and $p^{*}$ is projected to coincide with $r$ (as $w = r - p^{*}$), all points in $P^{*}$ will lie on the same $(k-1)$-dimensional flat $F'_1$ after projection. Also, since $w$ is orthogonal to $A$, there exists an affine map from $Q_{-}^{*} \cup Q_{+}^{*}$ to $P^{*}$ (after projection). Thus, we obtain that $\CH(P'_{-}) \cap \CH(P'_{+}) \neq \emptyset$; in particular, the convex hulls must intersect on $F'_1$. By Fact~\ref{fac:convexlinsep} this implies that $a_1$ is not strictly linearly separable after projection. 
\end{proof}

Although the projection of Theorem~\ref{thm:linsep} introduces a degeneracy (specifically, the points in $P^{*}$ are no longer in general position after projection), we show in Appendix~\ref{app:compdeg} how to eliminate this degeneracy without changing the result. 

\mypar{Computation}
The proof of Theorem~\ref{thm:linsep} is constructive and implies an efficient algorithm to compute the desired projection. Most steps in the construction involve simple linear algebra operations, like projections and intersecting flats (Gaussian elimination), which can easily be computed in polynomial time. The only nontrivial computational step is the application of Theorem~\ref{thm:kirchberger}, for which the proof is also constructive. If a point $x \in \CH(P) \cap \CH(Q)$ is given along with the coefficients for the convex combination, then we can simply obtain $P^{*}$ and $Q^{*}$ by repeatedly solving a linear system of equations and eliminating a point. Note that the linear system needs to involve only $d+3$ points (arbitrarily chosen), so the linear system of equations can be solved in $O(d^3)$ time, and we can eliminate a point and update the coefficients in the same amount of time. Thus, we can compute $P^{*}$ and $Q^{*}$ in $O(n d^3)$ time, where $n = |P| + |Q|$ (similar arguments were used in~\cite{DBLP:conf/soda/MeunierMSS17}). If we are not given a point in $x \in \CH(P) \cap \CH(Q)$ along with the coefficients for the convex combination, then this must be computed first. This can be computed efficiently using linear programming.

The proof of Theorem~\ref{thm:linsep} suggests how to check, if $P$ does not use all $2^k$ labels, if there exists a separability preserving projection that eliminates the linear separability of $a_1$: If we can find a set of separating hyperplanes $H_2, \ldots, H_k$ such that $\CH(P_{-})$ and $\CH(P_{+})$ intersect \emph{after} projecting them orthogonally onto the space spanned by the normals of $H_2, \ldots, H_k$, then the remainder of the proof holds. However, finding such suitable separating hyperplanes might be computationally hard in general.

\section{Generalized separability}\label{sec:GenSep}

In this section we consider a generalization of linear separability for classification. One approach to achieve more complicated classification boundaries is to use clustering: the label of a point is  determined by the label of the ``nearest'' cluster. If we use more than one cluster per class, then the resulting classification is more expressive than classification by linear separation. This approach is also strongly related to nearest-neighbor classification, another common machine learning technique: the points decompose the space into convex subsets, each of which is associated with exactly one point; given enough clusters, we can thus exactly capture this behavior. But even with few clusters (convex sets), it may be possible to reasonably approximate the decomposition by using a single cluster to capture the same of many points with the same label. Hence, Our generalized definition of separability is inspired by such clustering-based classifications, with convex sets modeling the clusters. 

\begin{figure}[b]
    \centering
    \includegraphics{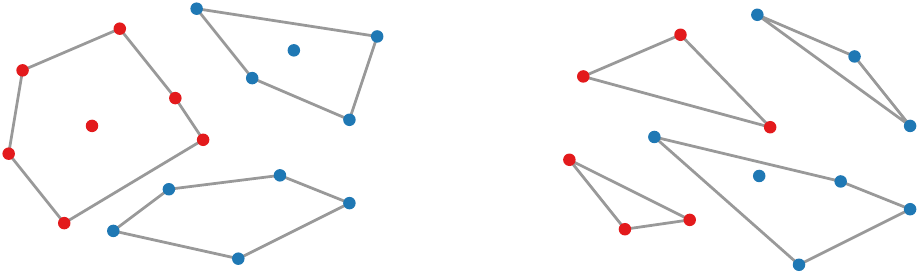}
    \caption{Left: two point sets $P$ (red) and $Q$ (blue) that are $(1,2)$-separable, but not linearly separable. Right: two point sets that are $(2,2)$-separable, but not $(1, x)$-separable for any value of $x$.}
    \label{fig:bcseparable}
\end{figure}

%
Let $P$ and $Q$ be two point sets in $\Reals^d$. We say that $P$ and $Q$ are \emph{$(b, c)$-separable} if there exist $b$ convex sets $S_1, \ldots, S_b$ and $c$ convex sets $T_1, \ldots, T_c$ such that for every point $p \in P$ we have that $p \in S = \bigcup_i S_i$, for every point $q \in Q$ we have that $q \in T = \bigcup_j T_j$, and that $S \cap T = \emptyset$ (see Figure~\ref{fig:bcseparable}). We can assume that $b \leq c$. Furthermore, we generally assume w.l.o.g.\ that any convex set $S_i$ is the convex hull of its contained points. It is easy to see that linear separability and $(1, 1)$-separability are equivalent. 




Given a point set $P$ along with $k$ properties $a_1, \ldots, a_k$, the goal is now to compute a separation preserving projection to a point set $P'$ such that $a_1$ is not $(b,c)$-separable in $P'$. We again assume that all $k$ properties are strictly linearly separable in $P$. To achieve this goal, we may need to project along multiple vectors $w_1, \ldots, w_r$. As mentioned in Section~\ref{sec:introduction}, we assume that $\{w_j\}_{j=1}^r$ form an orthonormal system and that we can compute the projected points as $p'_i = p_i - \sum_{j=1}^r (w_j \cdot p_i) w_j$. 

To extend Theorem~\ref{thm:linsep} to $(b, c)$-separability, recall the four main steps of the proof described before Theorem~\ref{thm:linsep}.
%
%
Step 3 is the most important. If $a_1$ was not linearly separable in $A$, then not even multiple separation preserving projections can eliminate the linear separability of $a_1$. In that sense, $A$ is the ``worst we can do'' with separation preserving projections. Step 3 is actually exploiting a Helly-type property~\cite{DBLP:reference/cg/Wenger04} for linear separability: If two sets of points $P$ and $Q$ are not linearly separable, then there exist small subsets $P^{*} \subseteq P$ and $Q^{*} \subseteq Q$ such that $P^{*}$ and $Q^{*}$ are not linearly separable (Theorem~\ref{thm:kirchberger}). Hence, if we use a different type of separability that also has a Helly-type property, then we may be able to use the same approach as for linear separability. Generally speaking, let $F(P, Q)$ be a predicate that determines if point sets $P, Q \subset \Reals^d$ are ``separable'' (for some arbitrary definition of separable)\footnote{We assume that $F$ is defined independently from the dimensionality of $P$ and $Q$ (like $(b, c)$-separability). We do require that $P$ and $Q$ are embedded in the same space.}. If, in the case that $F(P,Q)$ does not hold, there exist small (bounded by a constant) subsets $P^{*} \subseteq P$ and $Q^{*} \subseteq Q$ such that $F(P^{*}, Q^{*})$ also does not hold, then $F$ has the \emph{Helly-type property}. The worst-case size of $|P^{*}|+|Q^{*}|$ often depends on the number of dimensions $d$ of $P$ and $Q$, and is referred to as the \emph{Helly number} $m_F(d)$ of $F$. For technical reasons, we will require the following three natural conditions on $F$:

\begin{enumerate}
    \item If $F(P, Q)$ does not hold, then $F(P', Q')$ does not hold, where $P'$ and $Q'$ are obtained by projecting $P$ and $Q$ along a single unit vector, respectively.
    \item If $P' \subseteq P$ and $Q' \subseteq Q$, then $F(P, Q)$ implies $F(P', Q')$.
    \item If $\mathcal{A}$ is an affine map, then $F(P, Q)$ holds if and only if $F(\mathcal{A}(P), \mathcal{A}(Q))$ holds.
\end{enumerate}

We call a separation predicate $F$ \emph{well-behaved} if it satisfies these conditions. It is easy to see that $(b, c)$-separability is well-behaved. For Condition 1, note that any collection of convex sets for $P'$ and $Q'$ can easily be extended along the projection vector for $P$ and~$Q$ without introducing an overlap between $S$ and $T$. Condition 2 also holds, since we can simply use the same covering sets. Finally, Condition 3 holds since affine transformations preserve convexity. We summarize this generalization in the following generic theorem. 

\begin{restatable}{theorem}{Helly}
\label{thm:Helly}
Let $P$ be a point set in $\Reals^d$ with $k$ $(d \geq k)$ properties $a_1, \ldots, a_k$ and let $F$ be a well-behaved separation predicate in $\Reals^d$. Either we can use at most $\min(m_F(k-1) - k, d-k+1)$ separation preserving projections to eliminate $F(P_{-}, P_{+})$, or this cannot be achieved with any number of separation preserving projections. 
\end{restatable}
\begin{proof}
Following the proof of Theorem~\ref{thm:linsep}, we first orthogonally project the points in $P$ onto the $(k-1)$-dimensional linear subspace $A$ that is spanned by the normals $v_2, \ldots, v_k$ of the separating hyperplanes $H_2, \ldots, H_k$ of the properties $a_2, \ldots, a_k$. Let $T(p)$ be the resulting projected point for a point $p \in P$. Now define $Q_{-} = \{T(p)\mid p\in P_{-}\}$ and $Q_{+} = \{T(p)\mid p\in P_{+}\}$. If $F(Q_{-}, Q_{+})$ holds, then no sequence of separation preserving projections can eliminate the separability (as defined by $F$) of $a_1$, due to Condition 1 of a well-behaved separation predicate. Otherwise, we can find $Q_{-}^{*} \subseteq Q_{-}$ and $Q_{+}^{*} \subseteq Q_{+}$ such that $F(Q_{-}^{*}, Q_{+}^{*})$ does not hold, and $|Q_{-}^{*}| + |Q_{+}^{*}| \leq m_F(k-1)$. Let $P^{*} \subseteq P$ be the set of original points that map to $Q_{-}^{*} \cup Q_{+}^{*}$. The points in $P^{*}$ span a linear subspace $B$. Next, we construct an orthonormal basis $\{w_j\}_{j=1}^r$ for the set of vectors in $B$ that are orthogonal to $A$ (orthogonal to $v_2, \ldots, v_k$). Since $B$ has at most $m_F(k-1) - 1$ dimensions, and $A$ has $k-1$ dimensions, we conclude that the orthonormal basis contains $r \leq m_F(k-1) - 1 - (k-1) = m_F(k-1) - k$ vectors. We then choose to project $P$ along the vectors $w_1, \ldots, w_r$. Since every $w_j$ for $1 \leq j \leq r$ is orthogonal to $A$, these projections are all separation preserving. Furthermore, since we eliminate all vectors orthogonal to $A$ from $B$, there exists an affine map from $Q_{-}^{*} \cup Q_{+}^{*}$ to $P^{*}$ after projection. By using Condition 2 and Condition 3 of a well-behaved separation predicate, we can then conclude that $F(P'_{-}, P'_{+})$ does not hold. Alternatively, we can simply project $P$ to $A$, which requires $d - k + 1$ separation preserving projections. Hence, we need at most $\min(m_F(k-1) - k, d-k+1)$ projections.
\end{proof}

We now focus on $(b, c)$-separability for different values of $b$ and $c$. Unfortunately, not every form of $(b,c)$-separability has the Helly-type property.

\begin{restatable}{lemma}{notonetwo}
\label{lem:notonetwo}
In $d \geq 2$ dimensions, $(1,2)$-separability does not have the Helly-type property.
\end{restatable}
\begin{proof}
We prove the statement for $d=2$, which automatically implies it for $d > 2$. Consider a set of $n$ points $P = \{p_1, \ldots, p_n\}$ equally spaced on the unit circle, where $n$ is odd. For every point $p_i$ we can define a wedge $W_i$ formed between the rays from $p_i$ to the two opposite points on the circle (which are well defined, since $n$ is odd). By the Central Angle Theorem, the angle of this wedge is $\frac{\pi}{n}$. Furthermore, the distance of the rays to the origin is exactly $\sin\left(\frac{\pi}{n}\right)$. Now, for some $\epsilon > 0$ and for each point $p_i$, we add a point $q_i$ on the circle centered at the origin with radius $\sin\left(\frac{\pi}{n}\right) + \epsilon$, such that $q_i$ lies outside of $W_i$ to the left (counterclockwise). By construction there will also be a point $q_j$ to the right of $W_i$, added by the point $p_j$ that is the opposite point of $p_i$ on the right (clockwise) side. We choose $\epsilon$ small enough such that any wedge $W_i$ contains exactly $n-2$ points from $Q = \{q_1, \ldots, q_n\}$, having one point of $Q$ outside of $W_i$ on each side (see Figure~\ref{fig:notonetwo}).

Assume for the sake of contradiction that $P$ and $Q$ are $(1,2)$-separable. Since $Q \subset \CH(P)$, we must cover $Q$ with one set, and hence $S_1 = \CH(Q)$. Now consider $P_1 = T_1 \cap P$ and $P_2 = T_2 \cap P$. Since the line segments between a point $p_i \in P_1$ and its opposite points $p_j$ and $p_{j+1}$ intersect $\CH(Q)$, we get that $p_j$ and $p_{j+1}$ must both be in $P_2$. We can repeat this argument for all points $p_i$ to conclude that all pairs of consecutive points of $P$ must be in the same set ($P_1$ or $P_2$). Since not all points in $P$ can belong to the same set ($Q \subset \CH(P)$), we obtain a contradiction. Thus, $P$ and $Q$ are not $(1,2)$-separable.

Now consider removing a single point $p_i$ from $P$, and consider the line $\ell$ through the origin and $p_i$. The line $\ell$ splits $P\setminus\{p_i\}$ into two sets $P_1$ and $P_2$. It is easy to see that, if we pick $\epsilon$ small enough, $\CH(P_1)$ and $\CH(P_2)$ do not intersect $\CH(Q)$. Hence, $P\setminus\{p_i\}$ and $Q$ are $(1,2)$-separable. If we remove a single point $q_i$ from $Q$, then the line segment between $p_i$ and one of its opposite points $p_j$ does not intersect $\CH(Q\setminus\{q_i\})$. We can again split $P$ into $P_1$ and $P_2$ using the line $\ell$ through $p_i$ and $p_j$ (and shifted slightly towards the origin). Then it is again easy to see that, if we pick $\epsilon$ small enough, $\CH(P_1)$ and $\CH(P_2)$ do not intersect $\CH(Q\setminus\{q_i\})$. Hence, $P$ and $Q\setminus\{q_i\}$ are $(1,2)$-separable.

As a result, there exist no subsets of $P$ and $Q$ that are not $(1,2)$-separable. Thus, we get that the Helly number for $(1,2)$-separability is at least $|P| + |Q| = 2n$, and hence $(1,2)$-separability does not have the Helly-type property.
\end{proof}

\begin{figure}[t]
    \centering
    \includegraphics{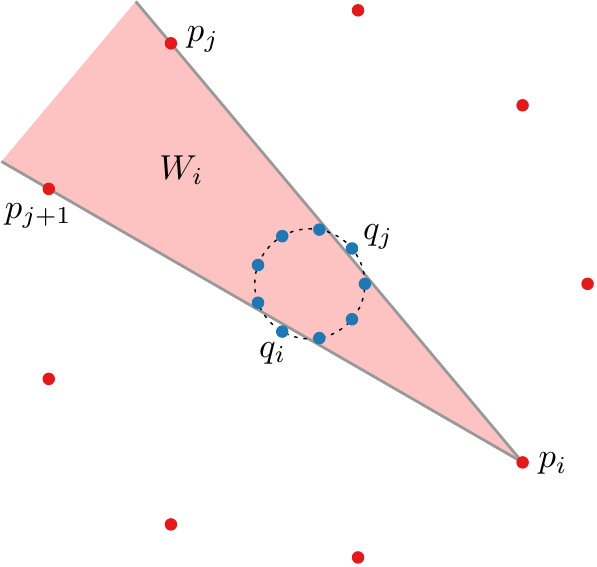}
    \caption{The construction for Lemma~\ref{lem:notonetwo} with $P$ in red and $Q$ in blue.}
    \label{fig:notonetwo}
\end{figure}

Hence we cannot apply Theorem~\ref{thm:Helly} to eliminate $(1,2)$-separability of $a_1$ in few separation preserving projections, if possible at all. However, this does not mean that it is not possible to provide this guarantee using different arguments. Nonetheless, we can use a similar construction as in the proof of Lemma~\ref{lem:notonetwo} (using many more dimensions) to show that many separation preserving projections are needed to eliminate $(1, 2)$-separability for $a_1$ (as many projections as needed to reach the $2$-dimensional construction in the proof of Lemma~\ref{lem:notonetwo}).

Next, we consider $(1, \infty)$-separability. This means that one of the point sets, say $P$, must be covered with one convex set, but we can use arbitrarily many convex sets to cover $Q$. Equivalently, $P$ and $Q$ are $(1, \infty)$-separable if $\CH(P) \cap Q = \emptyset$ or $P \cap \CH(Q) = \emptyset$.

\begin{restatable}{lemma}{oneinfinityyes}
\label{lem:one_infty_yes}
In $d \geq 1$ dimensions, $(1,\infty)$-separability has the Helly-type property with Helly number $2d+2$.
\end{restatable}
\begin{proof}
Let $P$ and $Q$ be point sets in $\Reals^d$ such that $P$ and $Q$ are not $(1, \infty)$-separable. Then there must be a point $p^{*} \in \CH(Q)$ and a point $q^{*} \in \CH(P)$. We can construct a star triangulation $\mathcal{T}(P)$ of $\CH(P)$ with $p^{*}$ as center (that is, all $d$-dimensional simplices have $p^{*}$ as a vertex) and a star triangulation $\mathcal{T}(Q)$ of $\CH(Q)$ with $q^{*}$ as center (see Figure~\ref{fig:one_infty_yes}). We identify the unique simplex $\sigma_P \in \mathcal{T}(P)$ that contains $q^{*}$, and similarly the unique simplex $\sigma_Q \in \mathcal{T}(Q)$ that contains $p^{*}$. Now let $P^{*} \subseteq P$ be the vertices of $\sigma_P$ and let $Q^{*} \subseteq Q$ be the vertices of $\sigma_Q$. Note that $p^{*} \in P^{*}$ and $q^{*} \in Q^{*}$. Then $P^{*}$ and $Q^{*}$ are not $(1, \infty)$-separable, since $q^{*} \in \CH(P^{*}) \cap Q^{*}$ and $p^{*} \in \CH(Q^{*}) \cap P^{*}$. Finally, since a $d$-dimensional simplex contains $d+1$ vertices, we obtain Helly number $2d + 2$.
\end{proof}

\begin{figure}[t]
    \centering
    \includegraphics{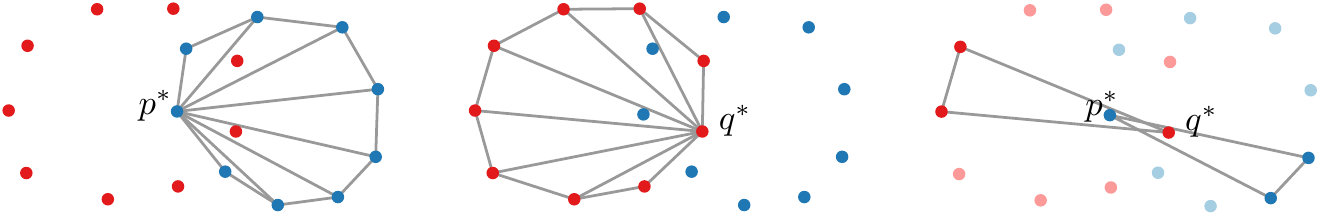}
    \caption{Lemma \ref{lem:one_infty_yes}: constructing a small point set that is not $(1, \infty)$-separable.}
    \label{fig:one_infty_yes}
\end{figure}

\begin{corollary}\label{cor:one_infty}
Let $P$ be a point set in $\Reals^d$ with $k$ $(d \geq k)$ properties $a_1, \ldots, a_k$. Either we can use at most $\min(k, d-k+1)$ separation preserving projections to eliminate $(1, \infty)$-separability of $a_1$, or this cannot be achieved with any number of separation preserving projections. 
\end{corollary}

It may initially seem counter-intuitive that $(1, \infty)$-separability has the Helly-type property (requiring only few projections to eliminate $(1, \infty)$-separability), while the strictly stronger $(1, 2)$-separability does not have the Helly-type property (and may require many projections to eliminate $(1,2)$-separability). Note however that Theorem~\ref{thm:Helly} includes the clause that it simply may not be possible to eliminate separability of $a_1$ via any number of separation preserving projections. This case occurs more often with $(1, \infty)$-separability than with $(1, 2)$-separability, which explains why we can provide better guarantees on the number of projections for a strictly weaker separability condition. 

We finally briefly consider $(2, \infty)$-separability in $\Reals^2$. Two point sets $P$ and $Q$ are not $(2, \infty)$-separable in $\Reals^2$ if we need at least three convex sets disjoint from $Q$ to cover $P$ (and vice versa). This implies that $\CH(P)$ must contain at least $3$ points of $Q$; if not, then we can draw a single line through all points in $Q \cap \CH(P)$ to separate $P$ into $P_1$ and $P_2$, and $\CH(P_1)$ and $\CH(P_2)$ both cover $P$ and are disjoint from $Q$. More generally, assume that we can cover $P$ with two sets $\CH(P_1)$ and $\CH(P_2)$ that are disjoint from $Q$, and let $\ell$ be a line that separates $\CH(P_1)$ and $\CH(P_2)$ (Fact~\ref{fac:convexlinsep}). Now consider the set of all triangles $\mathcal{T}_P$ that are formed by three points of $P$ such that a point of $Q$ is contained in the triangle. We must have that $\ell$ transverses (intersects) all triangles in $\mathcal{T}_P$, otherwise the triangle is contained in $P_1$ or $P_2$, and hence there is a point of $Q$ in either $\CH(P_1)$ or $\CH(P_2)$. Furthermore, if there is point $q \in Q$ contained in, say, $\CH(P_1)$, then there is also a triangle $\Delta \in \mathcal{T}_P$ in $P_1$ (Carath\'{e}odory's theorem), and hence $\ell$ does not intersect all triangles in $\mathcal{T}_P$. Thus, $P$ and $Q$ are $(2, \infty)$-separable (assuming we cover $P$ with $2$ convex sets) if and only if there exists a line $\ell$ that transverses $\mathcal{T}_P$. As a result, if we can show a Helly-type property for line transversals of triangles, then we also obtain a Helly-type property for $(2, \infty)$-separability. Unfortunately, there is no Helly-type property for line transversals of general sets of triangles~\cite{Lewis1980}. We leave it as an open question to determine if there exists a Helly-type property for line transversals of these special sets of triangles $\mathcal{T}_P$.

\section{Conclusion}\label{sec:conclusion}
We studied the use of projections for obstructing classification of high-dimensional Euclidean point data. Our results show that, if not all possible labels are present in the data, then it may not be possible to eliminate the linear separability of one property while preserving it for the other properties. This is not surprising if a property that we aim to keep is strongly correlated with the property we aim to hide. Nonetheless, one should be aware of this effect when employing projections in practice. When going beyond linear separability, we see that the number of projections required to hide a property increases significantly in theory, and we expect a similar effect when using, for example, neural networks for classification in practice. In other words, projecting a dataset once (or few times) may not be sufficient to hide a property from a smart classifier. Projection, as a linear transformation, can however be effective in eliminating certain linear relations in the data.

One potential direction of future work is to consider other separability predicates for labeled point sets, beyond linear separability and $(b,c)$-separability. Are there other types of separability that also have the Helly-type property used in Theorem~\ref{thm:Helly}? Or is there another way to show that few projections suffice to eliminate the separability of one of the properties? There are many other types of separability (for example, via boxes or spheres) for which this can be evaluated. 

In this paper we focused on eliminating bias based on one property (such as gender). Intersectionality posits that discrimination due to multiple properties should be considered in a holistic manner, instead of one property at a time. In fact, any one property might not be a cause for discrimination, but their combination is. The following challenge arises: say we used projection successfully to eliminate the linear separability of gender. However, if we now restrict the data to one particular sub-class, for example black people, then the linear separability of gender might still be preserved within this subclass and hence discrimination against black women can still be possible. Under which conditions is it possible to eliminate the linear separability of one property not only in the full data, but also in specific (or all) subclasses? We leave this question as an open problem.

\bibliography{references}

\begin{thebibliography}{10}

\bibitem{DBLP:conf/sdm/AbbasiFSV19}
Mohsen Abbasi, Sorelle~A. Friedler, Carlos Scheidegger, and Suresh
  Venkatasubramanian.
\newblock Fairness in representation: quantifying stereotyping as a
  representational harm.
\newblock In {\em Proc. {SIAM} International Conference on Data Mining}, pages
  801--809, 2019.
\newblock \href {https://doi.org/10.1137/1.9781611975673.90}
  {\path{doi:10.1137/1.9781611975673.90}}.

\bibitem{DBLP:conf/aies/AminiSSBR19}
Alexander Amini, Ava~P. Soleimany, Wilko Schwarting, Sangeeta~N. Bhatia, and
  Daniela Rus.
\newblock Uncovering and mitigating algorithmic bias through learned latent
  structure.
\newblock In {\em Proc. {AAAI/ACM} Conference on AI, Ethics, and Society},
  pages 289--295, 2019.
\newblock \href {https://doi.org/10.1145/3306618.3314243}
  {\path{doi:10.1145/3306618.3314243}}.

\bibitem{DBLP:journals/corr/BerkHJJKMNR17}
Richard Berk, Hoda Heidari, Shahin Jabbari, Matthew Joseph, Michael~J. Kearns,
  Jamie Morgenstern, Seth Neel, and Aaron Roth.
\newblock A convex framework for fair regression.
\newblock {\em arXiv:1706.02409}, 2017.

\bibitem{Pattern-Recognition-Bishop}
Christopher~M. Bishop.
\newblock {\em Pattern Recognition and Machine Learning (Information Science
  and Statistics)}.
\newblock Springer-Verlag, 2006.

\bibitem{DBLP:conf/nips/BolukbasiCZSK16}
Tolga Bolukbasi, Kai{-}Wei Chang, James~Y. Zou, Venkatesh Saligrama, and
  Adam~Tauman Kalai.
\newblock Man is to computer programmer as woman is to homemaker? {D}ebiasing
  word embeddings.
\newblock In {\em Advances in Neural Information Processing Systems 29: Annual
  Conference on Neural Information Processing Systems}, pages 4349--4357, 2016.

\bibitem{DBLP:conf/naacl/BordiaB19}
Shikha Bordia and Samuel~R. Bowman.
\newblock Identifying and reducing gender bias in word-level language models.
\newblock In {\em Proc. Conference of the North American Chapter of the
  Association for Computational Linguistics: Human Language Technologies},
  pages 7--15, 2019.
\newblock \href {https://doi.org/10.18653/v1/n19-3002}
  {\path{doi:10.18653/v1/n19-3002}}.

\bibitem{DBLP:conf/icml/BrunetAAZ19}
Marc~E. Brunet, Colleen~A. Houlihan, Ashton Anderson, and Richard~S. Zemel.
\newblock Understanding the origins of bias in word embeddings.
\newblock In {\em Proc. 36th International Conference on Machine Learning},
  volume~97, pages 803--811, 2019.

\bibitem{Caliskan183}
Aylin Caliskan, Joanna~J. Bryson, and Arvind Narayanan.
\newblock Semantics derived automatically from language corpora contain
  human-like biases.
\newblock {\em Science}, 356(6334):183--186, 2017.
\newblock \href {https://doi.org/10.1126/science.aal4230}
  {\path{doi:10.1126/science.aal4230}}.

\bibitem{Dev_Li_Phillips_Srikumar_2020}
Sunipa Dev, Tao Li, Jeff~M. Phillips, and Vivek Srikumar.
\newblock On measuring and mitigating biased inferences of word embeddings.
\newblock In {\em Proc. AAAI Conference on Artificial Intelligence}, pages
  7659--7666, 2020.

\bibitem{pmlr-v89-dev19a}
Sunipa Dev and Jeff Phillips.
\newblock Attenuating bias in word vectors.
\newblock In {\em Proc. Machine Learning Research}, pages 879--887, 2019.

\bibitem{edwards2016censoring}
Harrison Edwards and Amos Storkey.
\newblock Censoring representations with an adversary.
\newblock {\em arXiv:1511.05897}, 2016.

\bibitem{disparate}
Michael Feldman, Sorelle~A. Friedler, John Moeller, Carlos Scheidegger, and
  Suresh Venkatasubramanian.
\newblock Certifying and removing disparate impact.
\newblock In {\em Proc. 21st ACM SIGKDD International Conference on Knowledge
  Discovery and Data Mining}, page 259–268, 2015.

\bibitem{DBLP:conf/nips/HardtPNS16}
Moritz Hardt, Eric Price, and Nati Srebro.
\newblock Equality of opportunity in supervised learning.
\newblock In {\em Advances in Neural Information Processing Systems 29: Annual
  Conference on Neural Information Processing Systems}, pages 3315--3323, 2016.

\bibitem{Yuzi-AMC2020}
Yuzi He, Keith Burghardt, and Kristina Lerman.
\newblock A geometric solution to fair representations.
\newblock In {\em Proc. AAAI/ACM Conference on AI, Ethics, and Society}, page
  279–285, 2020.
\newblock \href {https://doi.org/10.1145/3375627.3375864}
  {\path{doi:10.1145/3375627.3375864}}.

\bibitem{DBLP:journals/kais/KamiranC11}
Faisal Kamiran and Toon Calders.
\newblock Data preprocessing techniques for classification without
  discrimination.
\newblock {\em Knowledge and Information Systems}, 33:1--33, 2011.
\newblock \href {https://doi.org/10.1007/s10115-011-0463-8}
  {\path{doi:10.1007/s10115-011-0463-8}}.

\bibitem{kirchberger1903}
Paul Kirchberger.
\newblock {\"U}ber {T}chebychefsche {A}nn\"aherungsmethoden.
\newblock {\em Mathematische Annalen}, 57:509--540, 1903.
\newblock \href {https://doi.org/10.1007/BF01445182}
  {\path{doi:10.1007/BF01445182}}.

\bibitem{Lewis1980}
Ted Lewis.
\newblock Two counterexamples concerning transversals for convex subsets of the
  plane.
\newblock {\em Geometriae Dedicata}, 9:461--465, 1980.
\newblock \href {https://doi.org/10.1007/BF00181561}
  {\path{doi:10.1007/BF00181561}}.

\bibitem{pmlr-v80-madras18a}
David Madras, Elliot Creager, Toniann Pitassi, and Richard Zemel.
\newblock Learning adversarially fair and transferable representations.
\newblock In {\em Proc. 35th International Conference on Machine Learning},
  pages 3384--3393, 2018.

\bibitem{DBLP:journals/corr/abs-1908-09635}
Ninareh Mehrabi, Fred Morstatter, Nripsuta Saxena, Kristina Lerman, and Aram
  Galstyan.
\newblock A survey on bias and fairness in machine learning.
\newblock {\em arXiv:1908.09635}, 2019.

\bibitem{DBLP:conf/soda/MeunierMSS17}
Fr{\'{e}}d{\'{e}}ric Meunier, Wolfgang Mulzer, Pauline Sarrabezolles, and
  Yannik Stein.
\newblock The rainbow at the end of the line - {A} {PPAD} formulation of the
  colorful {C}arath{\'{e}}odory theorem with applications.
\newblock In {\em Proc. 28th {ACM-SIAM} Symposium on Discrete Algorithms},
  pages 1342--1351, 2017.
\newblock \href {https://doi.org/10.1137/1.9781611974782.87}
  {\path{doi:10.1137/1.9781611974782.87}}.

\bibitem{Deng2012SupportVM}
Yingjie~Tian Naiyang~Deng and Chunhua Zhang.
\newblock {\em Support Vector Machines: Optimization Based Theory, Algorithms,
  and Extensions}.
\newblock CRC Press, 2012.

\bibitem{DBLP:conf/acl/RavfogelEGTG20}
Shauli Ravfogel, Yanai Elazar, Hila Gonen, Michael Twiton, and Yoav Goldberg.
\newblock Null it out: Guarding protected attributes by iterative nullspace
  projection.
\newblock In {\em Proc. 58th Annual Meeting of the Association for
  Computational Linguistics}, pages 7237--7256, 2020.

\bibitem{skeem2016risk}
Jennifer~L. Skeem and Christopher~T. Lowenkamp.
\newblock Risk, race, and recidivism: Predictive bias and disparate impact.
\newblock {\em Criminology}, 54(4):680--712, 2016.
\newblock \href {https://doi.org/10.1111/1745-9125.12123}
  {\path{doi:10.1111/1745-9125.12123}}.

\bibitem{DBLP:reference/cg/Wenger04}
Rephael Wenger.
\newblock Helly-type theorems and geometric transversals.
\newblock In Jacob~E. Goodman and Joseph O'Rourke, editors, {\em Handbook of
  Discrete and Computational Geometry, second edition}, pages 73--96. Chapman
  and Hall/CRC, 2004.
\newblock \href {https://doi.org/10.1201/9781420035315.ch4}
  {\path{doi:10.1201/9781420035315.ch4}}.

\bibitem{DBLP:conf/aistats/ZafarVGG17}
Muhammad~Bilal Zafar, Isabel Valera, Manuel~G. Rodriguez, and Krishna~P.
  Gummadi.
\newblock Fairness constraints: Mechanisms for fair classification.
\newblock In {\em Proc. 20th International Conference on Artificial
  Intelligence and Statistics}, volume~54, pages 962--970, 2017.

\bibitem{pmlr-v28-zemel13}
Rich Zemel, Yu~Wu, Kevin Swersky, Toni Pitassi, and Cynthia Dwork.
\newblock Learning fair representations.
\newblock In {\em Proc. 30th International Conference on Machine Learning},
  pages 325--333, 2013.

\bibitem{DBLP:conf/aies/ZhangLM18}
Brian~Hu Zhang, Blake Lemoine, and Margaret Mitchell.
\newblock Mitigating unwanted biases with adversarial learning.
\newblock In {\em Proc. {AAAI/ACM} Conference on AI, Ethics, and Society},
  pages 335--340, 2018.
\newblock \href {https://doi.org/10.1145/3278721.3278779}
  {\path{doi:10.1145/3278721.3278779}}.

\bibitem{DBLP:conf/emnlp/ZhaoZLWC18}
Jieyu Zhao, Yichao Zhou, Zeyu Li, Wei Wang, and Kai~Wei Chang.
\newblock Learning gender-neutral word embeddings.
\newblock In {\em Proc. Conference on Empirical Methods in Natural Language
  Processing}, pages 4847--4853, 2018.

\end{thebibliography}

\newpage
\appendix

\section{Eliminating degeneracy}\label{app:compdeg}
The result of Theorem~\ref{thm:linsep} has one shortcoming: the resulting projected point set $P'$ is degenerate by construction and property $a_1$ may still be (non-strictly) linearly separable after projection. This is simply an artifact of the proof and can be avoided by slightly perturbing the projection vector $w$. The following lemma can be used to remedy this shortcoming. Here we again assume that, before projection, the point set $P$ and the separating hyperplanes are in general position, and hence the only degeneracy in $P'$ is the one introduced by construction.

\begin{restatable}{lemma}{projfix}
\label{lem:projfix}
Let $P$ and $Q$ be two point sets in $\Reals^d$ in general position and let $P'$ and $Q'$ be the point sets obtained by projecting $P$ and $Q$ along a vector $w$, respectively. If $\CH(P') \cap \CH(Q') \neq \emptyset$, then we can perturb $w$ to obtain projections $P''$ and $Q''$ such that $P''$ and $Q''$ are not linearly separable and $P'' \cup Q''$ is in general position. 
\end{restatable}
\begin{proof}
Let $P = \{p_1, \ldots, p_n\}$ and $Q = \{q_1, \ldots, q_m\}$, and similarly $P' = \{p'_1, \ldots, p'_n\}$ and $Q' = \{q'_1, \ldots, q'_m\}$. We may assume that $m + n \geq d+1$, for otherwise $P$ and $Q$ do not really span $\Reals^d$. Since $\CH(P') \cap \CH(Q') \neq \emptyset$, there exist coefficients $\lambda_i \geq 0$ ($1 \leq i \leq n$) and $\mu_j \geq 0$ ($1 \leq j \leq m$) such that $\sum_i \lambda_i = 1$, $\sum_j \mu_j = 1$, and $\sum_i \lambda_i p'_i = \sum_j \mu_j q'_j$. We can ignore some points with a zero coefficient so that we have exactly $d+1$ points left, and we assume in the remainder of this proof that $m + n = d+1$. Now assume w.l.o.g. that $\lambda_1 > 0$. We use the remaining points $(P' \cup Q')\setminus\{p'_1\}$ to set up a barycentric coordinate system for the points in $P' \cup Q'$. This has the advantage that only the coordinates of $p_1$ are affected when changing the projection vector $w$. Next, we slightly perturb the coefficients to obtain $\lambda'_i > 0$ ($1 \leq i \leq n$), $\mu'_j > 0$ ($1 \leq j \leq m$) with $\lambda'_1 = \lambda_1$, $\sum_i \lambda'_i = 1$ and $\sum_j \mu'_j = 1$ (this is clearly possible). There then exist a vector $v$ (in barycentric coordinates) and $\epsilon > 0$ ($\epsilon$ can be arbitrarily small by scaling the perturbation of the coefficients) such that $\epsilon v + \sum_i \lambda'_i p'_i = \sum_j \mu'_j q'_j$. Now consider the point $p_1^{\perp}$ which has the same barycentric coordinates as $p'_1$, but then with the barycentric coordinate system defined by $(P \cup Q)\setminus\{p_1\}$. Then, by Observation~\ref{obs:projectcoeff}, we must have that $p_1 - p_1^{\perp} = \alpha w$ for some constant $\alpha \neq 0$. Now we perturb $p_1^{\perp}$ to $p^{*}$ such that $p^{*}$ has the same barycentric coordinates as $p'_1 + (\epsilon/\lambda_1) v$, but then again with the barycentric coordinate system defined by $(P \cup Q)\setminus\{p_1\}$. Additionally, we perturb $w$ to $w' = p_1 - p^{*}$. Let $P'' = \{p''_1, \ldots, p''_n\}$ and $Q'' = \{q''_1, \ldots, q''_m\}$ be the point sets obtained by projecting $P$ and $Q$ along $w'$. We then have by construction that $\sum_i \lambda'_i p''_i = \sum_j \mu'_j q''_j$. Now assume for the sake of contradiction that $P''$ and $Q''$ are linearly separable by a hyperplane $H$. Then $H$ must contain $\CH(P'')\cap\CH(Q'')$ and, consequently, all points that have a nonzero coefficient for the convex combination of a point $x \in \CH(P'')\cap\CH(Q'')$ (since all points of either $P''$ or $Q''$ lie on the same side of $H$). By construction there are $d+1$ of these points in $P'' \cup Q''$. Since $H$ is $(d-2)$-dimensional and we performed only a single projection, this also implies that there were $d+1$ points on a $(d-1)$-dimensional hyperplane in $P \cup Q$. This contradicts the assumption that $P \cup Q$ is in general position. Finally, since $\CH(P'')$ and $\CH(Q'')$ are not interior disjoint by Fact~\ref{fac:convexlinsep}, this property cannot be broken by slightly perturbing the projection vector $w'$. Thus, we can also ensure that $P'' \cup Q''$ is in general position.
\end{proof}

\section{Maximizing inseparability}\label{app:inseparable}
In this section we consider the problem of not only eliminating the linear separability of $a_1$, but additionally to maximize the ``linear inseparability'' (or overlap) of $a_1$ after projection. For that we need to define the overlap between two point sets $P$ and $Q$. For a unit vector $v$, consider the intervals $I_P(v) = \CH(\{v \cdot p_i \mid p_i \in P\}$ and $I_Q(v) = \CH(\{v \cdot q_i \mid q_i \in Q\}$. We can then define the overlap between $P$ and $Q$ along $v$ as the length of $I_P(v) \cap I_Q(v)$. Alternatively, we can define the overlap along $v$ with the cost function used by soft-margin SVMs, which is designed for data that is not linearly separable (see~\cite{Pattern-Recognition-Bishop} for more details). The overlap between two point sets $P$ and $Q$ is then defined as the minimum overlap over all (unit) vectors $v$. More precisely, for a given (projected) point set $P$, along with (implicit) property $a_1$, we use the function $g(P, v)$ to describe the overlap of $a_1$ along the vector $v$, and we refer to $g$ as the \emph{overlap function}. The overlap of $a_1$ is then defined as $\min_v g(P, v)$. Our goal is to find the projection that maximizes this overlap after projection. More precisely, if we use $P' = \pi_w(P)$ to denote the projection of a point set $P$ along the unit vector $w$, then the goal is to maximize the function $f(P, w) = \min_v g(\pi_w(P), v)$ over all separability/separation preserving projection vectors $w$. We consider the following two overlap functions $g(P, v)$ (although other options are possible):

\begin{description}
\item[Interval] For a point set $P$ and unit vector $v$, let $I_{-} = \CH(\{v \cdot p_i \mid p_i \in P_{-}\})$ and $I_{+} = \CH(\{v \cdot p_i \mid p_i \in P_{+}\})$. Then $g_{\text{int}}(P, v) = |I_{-} \cap I_{+}|$. 
\item[SVM] The goal of the soft-margin SVM optimization is to minimize $g_{\text{svm}}(P, v) = \lambda \|v\|^2 + \frac{1}{n}\sum_{i = 1}^n \max(0, 1 - a_1(p_i) (v \cdot p_i - b))$. Note that $g_{\text{svm}}$ also requires a parameter $b \in \Reals$, but we will often omit that dependence (we can assume that the overlap function minimizes over all $b \in \Reals$). Furthermore, $\lambda > 0$ is a parameter that can be set for $g_{\text{svm}}$. Finally, note that $v$ does not need to be a unit vector.
\end{description}

We first consider the variant of the problem that aims to find the optimal separability preserving projection. The vector $v$ that minimizes $g(P, v)$ for a given point set is typically computed using convex programming (in particular for SVMs, see~\cite{Deng2012SupportVM}). Note that convex programming heavily relies on the fact that there exists only one local optimum (which hence must be the global optimum). We show that, for the problem of finding the optimal separability preserving projection, there may be multiple local optima for $f(P, w)$. This eliminates the hope of finding a convex programming formulation for this problem.
\begin{theorem}\label{thm:two_local_optima}
There exists a point set $P$ in $\Reals^3$ with $2$ properties $a_1, a_2$ such that $f(P, w)$ with $g = g_{\text{svm}}$ has two local maxima when restricted to all separability preserving projection vectors $w$.  
\end{theorem}

\begin{proof}
The set $P$ mostly consists of the vertices of a unit cube with side lengths $2$ centered at the origin. We also add an extra point $p^{*} = (1-\epsilon, 1-\epsilon, 1)$ for some $\epsilon > 0$ (a point slightly moved inward from the point $p_8 = (1,1,1)$). Thus, $P$ consists of nine points $\{p_1, \ldots, p_8, p^{*}\}$. For property $a_1$ we choose that $a_1(p) = z(p)$ for all $p = (x(p), y(p), z(p)) \in P$. For property $a_2$ we have that $a_2(p) = a_1(p)$ for all $p \in P\setminus\{p_8\}$, and $a_2(p_8) = -1$. Now we limit and encode the space of possible projection vectors $w$ to $\Reals^2$ as follows. For $w = (x(w), y(w), z(w))$ with $z(w) = 0$ it is clear that a projection along $w$ will keep $a_1$ linearly separated, so we may encode all possible projections as $(x, y) = (x(w)/z(w), y(w)/z(w))$. Now consider the effect of using a projection $(x, y)$ on $P$: we may assume that the x- and y-coordinates of points $p \in P$ with $z(p) = 1$ do not change and that for the other points we obtain a shifted square: $(x(p'), y(p')) = (x(p) + 2x, y(p) + 2y)$ for all $p \in P$ with $z(p) = -1$. Let $p_1 = (-1,-1,-1)$ such that $p'_1 = (-1 + 2x, -1 + 2y)$. Furthermore, let $A$ consist of the projections of all points $p \in P$ with $a_1(p) = 1$, and let $B$ consist of the projections of all points $p \in P$ with $a_2(p) = 1$. Note that $A$ forms a square and $B$ forms a square with one of the corners pushed inwards. By Fact~\ref{fac:convexlinsep}, a projection $(x,y)$ can only be separability preserving if $p'_1 \notin \CH(B)$. By the same observation, a projection $(x,y)$ with $x \geq 0$ and $y \geq 0$ preserves the linear separability of $a_1$ if $p'_1 \notin \CH(A)$. Thus, we require that $p'_1 = (-1 + 2x, -1 + 2y) \in \CH(A)\setminus\CH(B)$. Note that $\CH(A)\setminus\CH(B)$ is a thin and nonconvex shape. The same thus holds for the domain of the projections $(x, y)$ as shown in Figure~\ref{fig:counter_example}, and hence the optimization problem is not convex. Furthermore, by evaluating the overlap function $g_{\text{svm}}$ (using $\lambda = 10$) on this domain, we can see that there are two distinct local maxima: one close to $(0, 1)$ and one close to $(1, 0)$.
\end{proof}
\begin{figure}[t]
    \centering
    \includegraphics[scale=0.5]{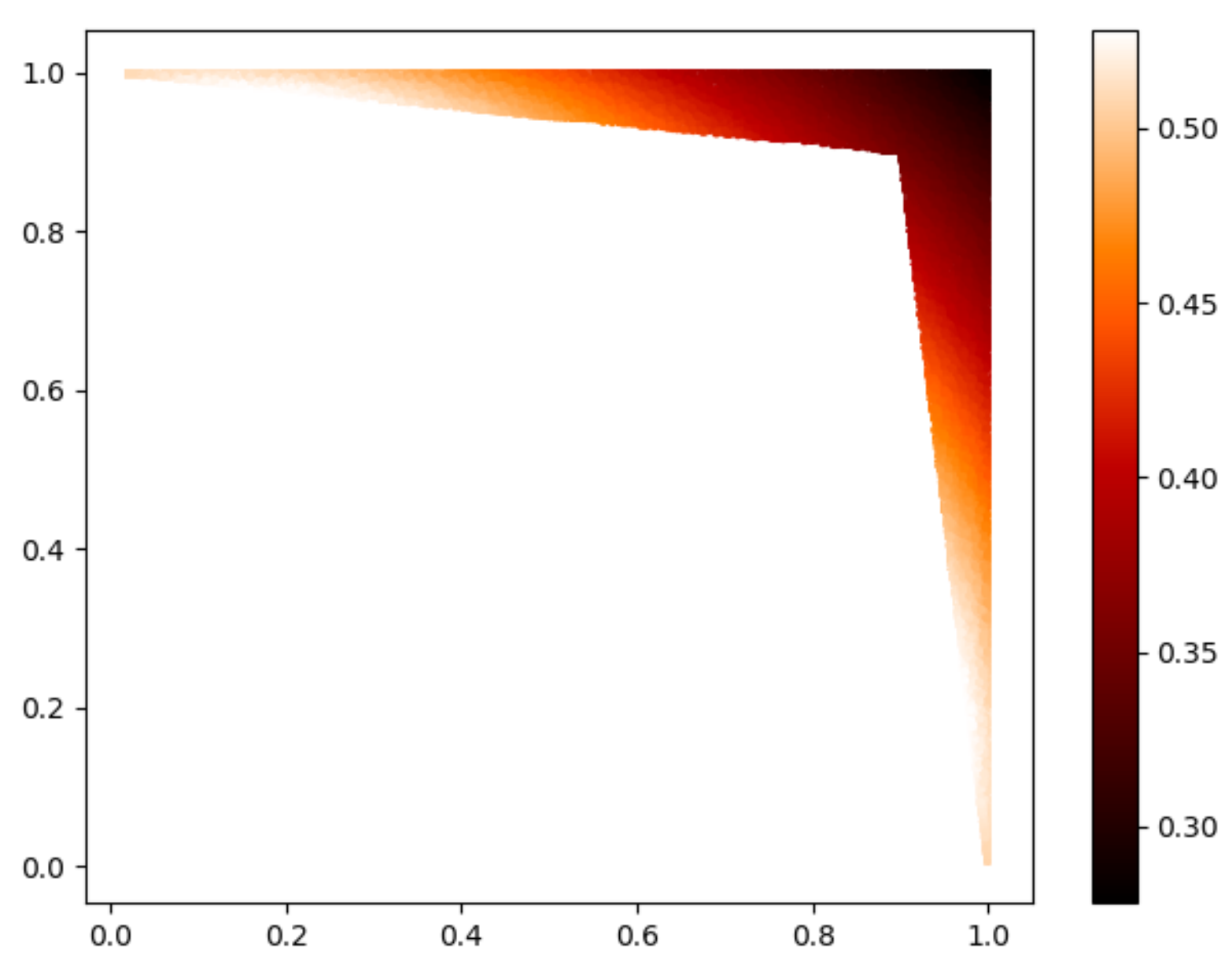}
    \caption{Illustration for Theorem~\ref{thm:two_local_optima}: The domain for projections $(x, y)$ with $\epsilon = 0.2$. Higher values in the overlap function are indicated with lighter colors. We can see two distinct local maxima.}
    \label{fig:counter_example}
\end{figure}
Theorem~\ref{thm:two_local_optima} demonstrates that the constraint on projections to be separability preserving is generally not convex. We now consider the special case that we have only one property $a_1$ (hence no separability preserving constraint), and analyze if we can then efficiently maximize $f(P, w)$. For that, we first put a restriction on the overlap function $g(P, v)$. We say that $g(P, v)$ is \emph{projectionable} if there exists a function $h\colon \Reals^n \times \Reals^d \rightarrow \Reals$ such that $g(P, v) = h(v \cdot P, v)$, where $v \cdot P = \{v \cdot p_i \mid p_i \in P\}$. In other words, $g$ should only depend on $P$ via the dot products of points in $P$ with $v$. Note that both the Interval and SVM overlap functions are indeed projectionable. For projectionable overlap functions $g$ we can redefine the optimization function $f$. In the following, let $v \perp w$ indicate that $v$ and $w$ are orthogonal.
\begin{restatable}{lemma}{projectionable}
\label{lem:projectionable}
If $g(P, v)$ is a projectionable overlap function, then\linebreak $\max_w \min_v g(\pi_w(P), v) = \max_w \min_{v \perp w} g(P, v)$ for any point set $P \subset \Reals^d$. 
\end{restatable}
\begin{proof}
We will treat both $w$ and $v$ as a vector in $\Reals^d$. 
Since $g$ is projectionable, there exists an equivalent function $h$ that depends on $v$ and the dot products between $v$ and points in $P$. Now assume that $(v \cdot w) = 0$. Then we get
\begin{align*}
g(\pi_w(P), v) &= h(v \cdot \pi_w(P), v) \\
&= h(\{v \cdot (p_i - (w \cdot p_i) w) \mid p_i \in P\}, v)\\
&= h(\{v \cdot p_i \mid p_i \in P\}, v)\\
&= g(P, v).
\end{align*}
Since the vector $v$ that minimizes $g(\pi_w(P), v)$ must be perpendicular to $w$, we obtain the desired equality.
\end{proof}
In the following we assume that the overlap function $g$ is projectionable. Hence, by Lemma~\ref{lem:projectionable}, we can rewrite $f$ as $f(P, w) = \min_{v \perp w} g(P, v)$. This has the advantage that we can keep the point set $P$ fixed while optimizing for $w$. We now aim to link properties of $g$ to properties of $f$. As already discussed earlier, we can often find the vector $v$ that minimizes $g(P, v)$ using convex programming. This implies that $g$ has only one local minimum (for fixed $P$). We now use this fact to show that $f$ has only one local maximum.
\begin{figure}
    \centering
    \includegraphics{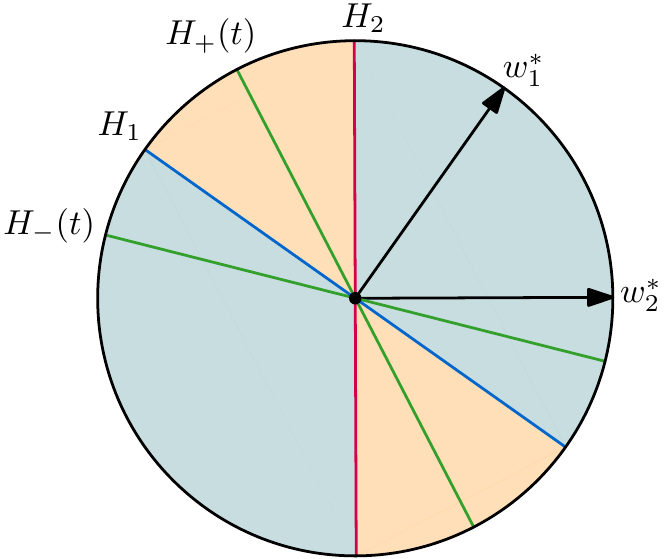}
    \caption{Illustration for Theorem~\ref{thm:one_local_max} with $d=2$: $H_{+}(t)$ can only intersect $R_{00} \cup R_{11}$ (orange, clipped to the unit disk) and $H_{-}(t)$ can only intersect $R_{01} \cup R_{10}$ (blue).}
    \label{fig:one_local_max}
\end{figure}
\begin{theorem}\label{thm:one_local_max}
If a function $g(P,v)$ has one local minimum for fixed $P$, then $f(P, w) = \min_{v \perp w} g(P, v)$ has one local maximum for fixed $P$.
\end{theorem}
\begin{proof}
Note that $w \in \mathcal{S}^{d-1}$, where $\mathcal{S}^{d-1}$ is the unit $(d-1)$-sphere, and $f(P, w) = f(P, -w)$, so we will treat $w$ and $-w$ as equivalent. Similarly, if the local minimum of $g(P, v)$ is at $v = v^{*}$, then $v = -v^{*}$ may also be a local minimum, and together they will be counted as a single local minimum. For the sake of contradiction, assume that $f(P, w)$ has two distinct local maxima, one at $w = w_1^{*}$ and one at $w = w_2^{*}$ (and also at $w = -w_1^{*}$ and $w = -w_2^{*}$). We do not require that $w_1^{*}$ and $w_2^{*}$ are strict local maxima, but we do require that there exists no path $\gamma\colon [0, 1] \rightarrow \mathcal{S}^{d-1}$ with $\gamma(0) = w_1^{*}$ and $\gamma(1) = \pm w_2^{*}$ such that $f(P, \gamma(t)) \geq \min(f(P, \gamma(0)), f(P, \gamma(1)))$ for all $0 \leq t \leq 1$. Now consider the hyperplanes $H_1 = \{v \mid (v \cdot w_1^{*}) = 0\}$ and $H_2 = \{v \mid (v \cdot w_2^{*}) = 0\}$. Furthermore, let $\gamma_{+}\colon [0, 1] \rightarrow \mathcal{S}^{d-1}$ denote the shortest (hyper)spherical interpolation from $w_1^{*}$ to $w_2^{*}$, and let $\gamma_{-}\colon [0, 1] \rightarrow \mathcal{S}^{d-1}$ denote the shortest (hyper)spherical interpolation from $w_1^{*}$ to $-w_2^{*}$. Note that $\gamma_{-}$ and $\gamma_{+}$ are unique, since both $w_1^{*}$ to $w_2^{*}$ lie on a great circle on $\mathcal{S}^{d-1}$ and $w_1^{*} \neq -w_2^{*}$. The hyperplanes $H_1$ and $H_2$ split $\Reals^d$ into four parts (each hyperplane cuts $\Reals^d$ into two parts): $R_{00}, R_{01}, R_{10}$, and $R_{11}$. Now let $x^{*} = \min(f(P, w_1^{*}), f(P, w_2^{*}))$ and consider the sublevel set $S = \{v \mid g(P, v) < x^{*}\}$. By definition of $f$, $H_1$ and $H_2$ are disjoint from $S$. Now consider a vector $\gamma_{+}(t)$ for some $0 < t < 1$, and let $H_{+}(t) = \{v \mid (v \cdot \gamma_{+}(t)) = 0\}$ be the corresponding hyperplane. Similarly define $H_{-}(t)$ for $\gamma_{-}(t)$.
It is easy to see that $H_{+}(t)$ intersects either $R_{01} \cup R_{10}$ or $R_{00} \cup R_{11}$, but not both, and $H_{-}(t)$ intersects only the other region (see Figure~\ref{fig:one_local_max}). By assumption, there exist values $t_{-}^{*}$ and $t_{+}^{*}$ such that $f(P, \gamma_{-}(t_{-}^{*})) < x^{*}$ and $f(P, \gamma_{+}(t_{+}^{*})) < x^{*}$. Thus, by the definition of $f$, there must be two non-opposite regions, say $R_{00}$ and $R_{01}$, that contain a point in $S$. These points cannot be in the same connected component, as they are separated by either $H_1$ or $H_2$. Thus, $S$ has multiple (non-opposite) connected components, and hence $g(P, v)$ must have at least two local minima. This contradicts our assumption, and hence $f(P, w)$ can have at most one local maximum.
\end{proof}
Following Theorem~\ref{thm:one_local_max}, we can use a hill-climbing approach to find the optimal projection vector $w$, if there is only one property $a_1$. This same approach can be applied to find the optimal separation preserving projection for $k$ properties. In that case, the corresponding separating hyperplanes $H_2, \ldots, H_k$ each take away a degree of freedom, but otherwise do not bound the domain of $w$. More precisely, if there is only one property $a_1$, then $w \in \mathcal{S}^{d-1}$, where $\mathcal{S}^d$ is the unit $d$-sphere. If there are $k$ properties, then $w \in \mathcal{S}^{d-1} \cap H_2 \cap \ldots \cap H_k = \mathcal{S}^{d-k}$. This reduction in dimensionality of the domain of $w$ does not affect the proof of Theorem~\ref{thm:one_local_max}.


\section{Omitted proofs}\label{app:omitted}

\begin{figure}[b!]
    \centering
    \includegraphics{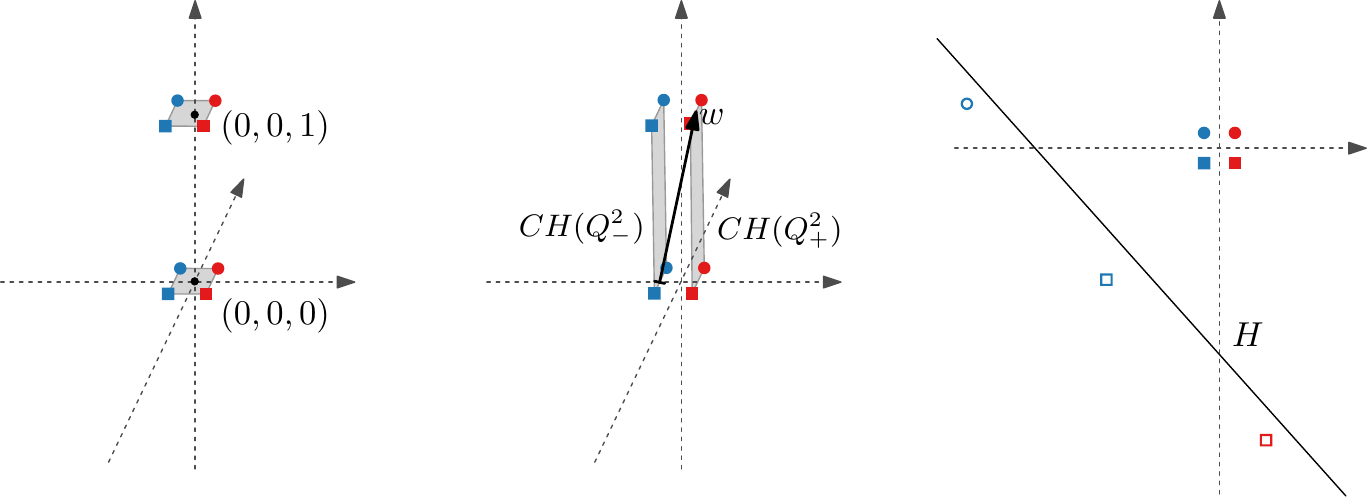}
    \caption{Illustration for Lemma~\ref{lem:not_all_labels} with $d = k = 3$ and properties fill ($a_1$), color ($a_2$), and shape ($a_3$). Left: $Q$ consisting of two copies of $C_{\epsilon}$. Middle: a separability preserving projection must be nearly orthogonal to the $(x, y)$-plane. Right: the flat $H$ separating property $a_1$.}
    \label{fig:not_all_lables}
\end{figure}

\notalllabels*

\begin{proof}
We first construct the point set $P$ for arbitrary $k$ and $d = k$. Consider the vertices of a $(k-1)$-dimensional hypercube $C_{\epsilon}$ with side length $\epsilon > 0$ centered at the origin, for which all nonzero coordinates lie in the first $k-1$ dimensions of $\Reals^d$. For each vertex $p$ of $C_{\epsilon}$, set the properties of $p$ based on its coordinates $(p^1, \ldots, p^d)$: $a_1(p) = 1$, and $a_i(p) = \sgn(p^{i-1})$ for $2 \leq i \leq k$, where $\sgn(x)$ is the sign function. Next, create a copy of $C_{\epsilon}$ (along with the assigned properties) and place it around the coordinate $(0, \ldots, 0, 1)$ (see Figure~\ref{fig:not_all_lables} left). Let the resulting point set be $Q$, and consider projecting $Q$ along a unit vector $w$. Let $w = (w^1, \ldots, w^d)$ and assume w.l.o.g. that $|w^1| \geq |w^i|$ for all $2 \leq i < k$. If $|w^1| > \epsilon$, then there always exists a line $\ell$ parallel to $w$ that intersects both $\CH(Q^2_{-})$ and $\CH(Q^2_{+})$. By Lemma~\ref{lem:sepafterproj} this would imply that $a_2$ is not strictly linearly separable after projection along $w$, so we may assume that $|w^1| \leq \epsilon$ for any separability preserving projection (see Figure~\ref{fig:not_all_lables} middle). 

Now consider a $(k-2)$-dimensional flat $H$ with the following properties: (1) it is not parallel to one of the first $k-1$ axes, (2) it lies in the first $k-1$ dimensions of $\Reals^d$ (the other coordinates are zero), and (3) the distance from the origin to $H$ is $1$ (see Figure~\ref{fig:not_all_lables} right). Consider the orthants of the $(k-1)$-dimensional subspace $A$ spanned by the first $k-1$ axes. Based on the labels of the vertices of $C_{\epsilon}$, each orthant is associated with a label for the properties $a_2, \ldots, a_k$. Due to Property (1), $H$ intersects all the first $k-1$ axes, either at the positive or the negative half-axis. Since there is exactly one orthant bounded by only the non-intersected half-axes, $H$ intersects exactly $2^{k-1}-1$ orthants. We now construct $P$ by extending $Q$ with an additional point in each of the intersected orthants, such that $H$ separates this point from the origin. The label of each such point $p$ has $a_1(p) = -1$ and is otherwise determined by the orthant. As a result, $P$ uses $2^k-1$ different labels. 

Let $v$ be the normal of $H$ in the $(k-1)$-dimensional subspace $A$. The margin for $P_{-}$ and $P_{+}$ along $v$ is at least $1 - k \epsilon$ (rough bound). For any separability preserving projection along unit vector $w$, we have that $|(w \cdot v)| \leq \epsilon$. Now consider any point $p \in P$ and its projection $p' = p - (w \cdot p) w$. We have that $(p' \cdot v) = (p \cdot v) - (w \cdot p) (w \cdot v) = (p \cdot v) \pm O(\epsilon)$, where we use the fact that $(w \cdot p) = O(1)$. Thus, the margin for property $a_1$ can be reduced by at most $O(\epsilon)$ by the projection, and hence the projection keeps $a_1$ strictly linearly separable if we choose $\epsilon$ small enough.

If $d > k$, then we can construct a simplex with side lengths $1$ in the last $d-k+1$ dimensions, and place a copy of $C_{\epsilon}$ around each of its vertices (for $d=k$ this simplex is simply an edge, as used above). With this construction we can still enforce w.l.o.g. that $|w^1| \leq \epsilon$ for any separability preserving projection along unit vector $w$, and the rest of the argument follows. 
\end{proof}

\end{document}